\documentclass[11pt,a4paper]{article}
\usepackage[utf8]{inputenc}

\usepackage{amsmath, amssymb, amsthm}
\usepackage{graphicx}
\usepackage{enumerate}
\usepackage{color}
\usepackage{subcaption}

\usepackage[affil-it]{authblk}
\author{Maël Dumas}
\author{Anthony Perez}
\affil{Univ. Orl\'eans, INSA Centre Val de Loire, LIFO EA 4022, F-45067 Orl\'eans, France}

\usepackage{fullpage}
\usepackage{amsthm}
\newtheorem{theorem}{Theorem}
\newtheorem{lemma}{Lemma}

\newtheorem{claim}{Claim}
\newtheorem{proposition}{Proposition}
\newtheorem{definition}{Definition}
\theoremstyle{definition}

\theoremstyle{theorem}
\newtheorem{observation}{Observation}

\newtheorem{polyrule}{Rule}

\usepackage{hyperref}
\hypersetup{
    colorlinks,
    linkcolor={red!50!black},
    citecolor={blue!50!black},
    urlcolor={blue!80!black}
}
\usepackage[ocgcolorlinks]{ocgx2}

\usepackage[capitalise]{cleveref}
\newenvironment{proofclaim}{
	\noindent \emph{Proof.}
}{%
	\hfill $\diamond$ \\
}

\crefname{claim}{Claim}{Claims}
\crefname{observation}{Observation}{Observations}
\crefname{polyrule}{Rule}{Rules}
\usepackage[defaultlines=3,all]{nowidow}

\usepackage[skins]{tcolorbox}
\tcbuselibrary{many}

\newtcolorbox{mypb}[2][]
{
    enhanced,
    boxed title style = {colframe=white},
    attach boxed title to top left={
        xshift=0.5cm,
        yshift= -3.5mm,     
    },
    top=4mm,
    coltitle=black,
    beforeafter skip=\baselineskip,
    colframe = lightgray,
    colback  = white,
    colbacktitle  = white,
    coltitle = black,  
    fonttitle = \scshape,
    titlerule = 0mm, 
    title    = {#2},
    #1
}

\newcommand{\Pb}[4]{%
    \begin{mypb}{#1}
       \textbf{\textsf{Input}}: #2%
       \par\noindent%
       \textbf{\textsf{#4}}: Does there exist #3?%
       \smallskip%
       \par\noindent%
    \end{mypb}
}

\renewcommand{\leq}{\leqslant}
\renewcommand{\geq}{\geqslant}

\def\eg{{\em e.g.}}
\def\etal{{\em et al.}}
\def\ie{{\em i.e.}~}

\def\busymb{\circ}

\def\VP{V_p}

\def\CL{C_<}
\def\CM{C_\busymb}
\def\CU{C_>}
\def\RL{R_<}
\def\RM{R_\busymb}
\def\RU{R_>}
\def\VM{v_\busymb}

\def\FM{F_\busymb}
\def\GM{G_\busymb}
\def\HM{H_\busymb}
\def\SM{S_\busymb}

\def\WC{comb}

\def\TPE{\textsc{Trivially Perfect Editing}}
\def\TPC{\textsc{Trivially Perfect Completion}}
\def\INST{$(G=(V,E),k)$}
\def\YINST{\textsc{Yes}-instance}

\def\PP{$r$-packing}

\def\DPT{Dumas \etal~\cite{DPT23}}

\newcommand{\SPP}[1]{$(#1)$-packing}

\newcommand{\BLOW}[3]{#1 (#3 \rightarrow #2)}
\newcommand{\size}[1]{\vert #1 \vert}
\newcommand{\OPT}[1]{#1^*}

\crefname{lemma}{Lemma}{Lemmata}

\title{An improved kernelization algorithm for \TPE}

\begin{document}
\maketitle

\begin{abstract}
    In the \TPE{} problem one 
    is given an undirected graph $G = (V,E)$ and an integer $k$ and seeks to 
    add or delete at most $k$ edges in $G$ to obtain a trivially perfect graph. 
    In a recent work, \DPT{} proved that this problem  
    admits a kernel with $O(k^3)$ vertices.
    This result heavily relies on the fact that the size of trivially 
    perfect modules can be bounded by $O(k^2)$ as shown by Drange and Pilipczuk~\cite{DP18}. 
    To obtain their cubic vertex-kernel, \DPT{} then showed that a more intricate structure, 
    so-called \emph{comb}, can be reduced to $O(k^2)$ vertices. 
    In this work we show that the bound can be improved to $O(k)$ for both aforementioned structures and thus obtain a kernel with $O(k^2)$ vertices. 
    Our approach relies on the straightforward yet powerful observation that any large enough structure contains unaffected vertices whose neighborhood remains unchanged by an editing of size $k$, implying strong structural properties. 
\end{abstract}

\section{Introduction}
\label{sec:intro}

In the \TPE{} problem one is given an undirected graph $G = (V,E)$ and an integer $k$ and seeks to \emph{edit} (add or delete) at most $k$ edges in $G$ so that the resulting graph  
is trivially perfect (\ie does not contain any cycle on four vertices nor path on four vertices as an induced subgraph). More formally we consider the following problem:

\Pb{Trivially Perfect Editing}%
{A graph $G=(V,E)$, a \emph{parameter} $k \in \mathbb{N}$}%
{a set $F \subseteq [V]^2$ of size at most $k$, such that the graph 
$H = (V, E \triangle F)$ is trivially perfect}
{Question}

Here $[V]^2$ denotes the set of all pairs of elements of $V$ and $E \triangle F = (E \cup F) \setminus (E \cap F)$ is the symmetric difference between sets $E$ and $F$. We define similarly the deletion (resp. completion) variant of the problem by only allowing to delete (resp. add) edges. 
Graph modification covers a broad range of well-studied problems that find applications 
in various areas. For instance, \TPE{} has been used to  define the community 
structure of complex networks by Nastos and Gao~\cite{NG13} and is closely related to the well-studied graph parameter tree-depth~\cite{Golumbic78,NdM15}.  
Theoretically, some of the earliest NP-Complete 
problems are graph modification problems~\cite{Karp72,GJ02}. Regarding edge (graph) modification problems, one of the most notable one is the  
\textsc{Minimum Fill-in} problem which aims at adding edges to a given graph to obtain a chordal graph  
(\ie a graph that does not contain any induced cycle of length at least $4$). 
 In a seminal result, 
Kaplan \etal~\cite{KST99} proved that \textsc{Mininum Fill-in} admits a parameterized 
algorithm as well as a kernel containing $O(k^3)$ vertices. This result was 
later improved to $O(k^2)$ vertices by Natanzon \etal~\cite{NSS00}. 
Parameterized complexity and kernelization algorithms 
provide a powerful theoretical framework to 
cope with decision problems.

\smallskip

\paragraph*{Parameterized complexity} 
A parameterized problem $\Pi$ is a language of $\Sigma^* \times \mathbb{N}$, 
where $\Sigma$ is a finite alphabet. An instance of a parameterized problem is a pair $(I,k)$ with $I \subseteq \Sigma^*$ 
and $k \in \mathbb{N}$, called the \emph{parameter}. A parameterized problem is said to be \emph{fixed-parameter tractable} 
if it can be decided in time $f(k) \cdot \size{I}^{O(1)}$. An equivalent definition of fixed-parameter tractability is the notion of \emph{kernelization}. Given an instance $(I,k)$ of a parameterized problem $\Pi$, 
a \emph{kernelization algorithm} for $\Pi$ (kernel for short) is a polynomial-time algorithm that outputs 
an equivalent instance $(I',k')$ of $\Pi$ 
such that $\size{I'} \leqslant h(k)$ for some function $h$ depending on the parameter only and $k' \leqslant k$. 
It is well-known that a parameterized problem is fixed-parameter tractable if and only if it admits a 
kernelization algorithm (see \eg{} \cite{FLS+19}). 
Problem $\Pi$ is said to admit a \emph{polynomial kernel} whenever $h$ is a polynomial. 

\paragraph*{Related work} 
Since the work of Kaplan \etal~\cite{KST99} many polynomial kernels for 
edge modification problems have been devised (see \eg~\cite{BP13,BPP10,DDL+22,Guo07,BBC+22,DP18,KU08,DPT21,CPT22}).
There is also evidence that  
under some reasonable theoretical complexity assumptions, some graph modification problems 
do not admit polynomial kernels~\cite{KW13,GHP+13,CC15,MS22}. We refer the reader to a recent comprehensive survey  
on kernelization for edge modification problems by Crespelle \etal~\cite{CDF+23}. 
The \TPE{} problem has been well-studied in the literature~\cite{BHS15,GHS+20,BHH+23,BBC+22,DP18,DPT23,LJY+15,NG13}. 
Recall that trivially perfect graphs are a subclass of chordal graphs that additionally do not contain any path on four vertices as an induced subgraph. These graphs are also known as \emph{quasi-threshold} graphs. 
We note here that while the NP-Completeness of completion and deletion toward trivially perfect graphs 
has been known for some time~\cite{Sharan02,BBD06}, the NP-Completeness of \TPE{} remained open until 
the work of Nastos and Gao~\cite{NG13}. 
Thanks to a result of Cai~\cite{Cai96} stating that graph modification toward \emph{any} graph class 
characterized by a finite set of forbidden induced subgraphs is fixed-parameter tractable, \TPE{} is  
fixed-parameter tractable.   
Regarding kernelization algorithms, Drange and Pilipczuk~\cite{DP18} provided a kernel 
containing $O(k^7)$ vertices, a result that was recently improved to $O(k^3)$ vertices by \DPT. 
These results also work for the deletion and completion variants. For the latter problem, a recent result by Bathie \etal~\cite{BBC+22} improves the bound to $O(k^2)$ vertices. 

\medskip

As part of the proof for the size of their cubic vertex-kernel, \DPT{} subsequently showed the following result. The structures used in \cref{thm:strategy} shall be defined later. 

\begin{theorem}[\cite{DPT23}]
\label{thm:strategy}
Let $(G,k)$ be an instance\footnote{As we shall see \cref{subsec:modules} the instance also needs to be further reduced under standard reduction rules.} of \TPE{} such that the sizes of its \emph{trivially perfect modules} 
and \emph{combs} are 
bounded by $p(k)$ and $c(k)$, respectively. If $(G,k)$ is a \YINST{} then $G$ has $O(k^2 + k \cdot (p(k)+c(k)))$ vertices. 
\end{theorem}

The proof of~\cite[Theorem 1]{DPT23} actually implies an $O(k^3 + k \cdot (p(k)+c(k)))$ bound and needs a small adjustment for \cref{thm:strategy} to hold. We give a detailed proof of \cref{thm:strategy} for the sake of completeness (\cref{subsec:bound}). 

\smallskip

The cubic vertex-kernel of \DPT{} relied on a result of Drange and Pilipczuk~\cite{DP18} that proved 
that $p \in O(k^2)$ and then used new  reduction rules implying that $c \in O(k^2)$. 

\paragraph*{Our contribution} 
We provide reduction rules and structural properties on trivially perfect graphs 
that will imply an $O(k)$ bound for both 
functions $p$ and $c$ of \cref{thm:strategy}. These new reduction rules allow us to prove the existence a quadratic vertex-kernel for \TPE{}.
To bound the size of trivially perfect modules by $O(k)$, we first reduce the ones that contain a large matching of non-edges with the use of a simple reduction rule. 
To bound the ones that do not contain such structures, we will rely on so-called \emph{combs}, introduced by \DPT. Combs correspond to parts of the graph that induce trivially perfect graphs (but not necessarily modules) with strong properties on their neighborhoods. They are composed of two main parts, called the \emph{shaft} and the \emph{teeth}, that will be independently reduced to a size linear in $k$. The reduction rule dealing with shafts will 
ultimately allow us to bound the size of trivially perfect modules with no large matching of non-edges. 
Our approach relies on the straightforward yet powerful observation that any large enough structure contains unaffected vertices whose neighborhood remains unchanged by an editing of size $k$. 
Finally, we note that our kernel works for both the deletion and completion variants of the problem. 

\paragraph*{Outline} 
\cref{sec:preliminaries} presents some 
preliminary notions and structural properties on (trivially perfect) graphs. \cref{sec:rules} describes known as well as our additional 
reduction rules to obtain the claimed kernelization algorithm while 
\cref{sec:deletion} explain why our kernel is safe  for 
the deletion variant of the problem. We conclude with some perspectives 
in \cref{sec:conclusion}.   

\section{Preliminaries}
\label{sec:preliminaries}

We consider simple, undirected graphs $G=(V,E)$ where $V$ denotes the vertex set of $G$ and $E \subseteq [V]^2$ its edge set. We will sometimes use 
$V(G)$ and $E(G)$ to clarify the context. 
The open (respectively closed) neighborhood of a vertex $u \in V$ is denoted by $N_G(u) = \{v \in V \mid \{u,v\} \in E\}$ 
(respectively $N_G[u] = N_G(u) \cup \{u\}$). Given a subset of vertices $S \subseteq V$ the neighborhood 
of $S$ is defined as 
$N_G(S) = \cup_{v \in S} N_G(v) \setminus S$. 
Similarly, given a vertex $u \in V$ and $S \subseteq V$ we 
let $N_S(u) = N_G(u) \cap S$. In all aforementioned cases 
we forget the subscript mentioning graph $G$ whenever the context is clear. 
Given a subset of vertices $S \subseteq V$ we denote by $G[S]$ the subgraph induced by $S$, that is $G[S] = (S,E_S)$ where $E_S = \{uv \in E:\ u \in S, v \in S\}$.   
In a slight abuse of notation,  
we use $G \setminus S$ to denote the induced subgraph $G[V \setminus S]$.
A \emph{connected component} is a maximal subset of vertices $S \subseteq V$ 
such that $G[S]$ is connected. 
A \emph{module} of $G$ is a set $M \subseteq V$ such that 
for all $u,v \in M$ it holds that $N(u) \setminus M = N(v) \setminus M$. Two vertices $u$ and $v$ are 
\emph{true twins} whenever $N[u] = N[v]$, and a \emph{critical clique} is a maximal set of true twins. 
A vertex $u \in V$ is \emph{universal} if $N_G[u] = V$.
The set of universal vertices forms a clique and is called the \emph{universal clique} of $G$. 
A graph is \emph{trivially perfect} if and only if it does not contain 
any $C_4$ (a cycle on $4$ vertices) nor $P_4$ (a path on $4$ vertices) 
as an induced subgraph. 
In the remainder of this section we describe  characterizations and 
structural properties of trivially perfect graphs. 
The first one relies on the well-known fact that 
any connected trivially perfect graph contains a 
universal vertex (see \eg~\cite{YCC96}). 

\begin{definition}[Universal clique decomposition,~\cite{DFP+15}]
\label{def:ucd}
A \emph{universal clique decomposition} (UCD) of a connected graph $G=(V,E)$ is a pair $\mathcal{T} = \big ( T= (V_T,E_T), \mathcal{B} = \{B_t\}_{t\in V_T} \big )$ where $T$ is a rooted tree and $\mathcal{B}$ is a partition of the vertex set $V$ into disjoint nonempty subsets such that:
\begin{itemize}
    \item if $\{v,w\} \in E$ and $v \in B_t$, $w \in B_s$ then $s$ and $t$ are on a path
    from a leaf to the root, with possibly $s=t$,
    \item for every node $t \in V_T$, the set $B_t$ of vertices is the universal clique of the induced subgraph $G[\bigcup_{s\in V(T_t)}B_s]$, where $T_t$ denotes the subtree of $T$ rooted at $t$.
\end{itemize}
\end{definition}

A simple way of understanding \cref{def:ucd} is to 
observe that such a decomposition can be obtained by 
removing the set $U$ of universal 
vertices of $G$ and then recursively repeating this process on every  
trivially perfect connected component of $G \setminus U$. %
Drange \etal~\cite{DFP+15} showed that a connected graph admits a UCD if and only if it is trivially perfect. 
Using the notion of UCD, \DPT{} proved the 
following characterization for trivially perfect graphs that will be heavily used in our reduction rules. A collection of subsets $\mathcal{F} \subseteq 2^U$ over some universe $U$ is a \emph{nested family} if $A \subseteq B$ or $B \subseteq A$ holds for any $A,B \in \mathcal{F}$. 

\begin{lemma}[\cite{DPT23}]
\label{lem:TP_recons}
Let $G=(V,E)$ be a graph, $S \subseteq V $ a maximal clique of $G$ and $\{K_1,...,K_r\}$ the set of connected components of $G\backslash S$. The graph $G$ is trivially perfect if and only if the following conditions are verified: 
\begin{enumerate}[(i)] 
    \item\label{it1:recons}  $G[S \cup K_i]$ is trivially perfect, $1\leq i \leq r$
    \item\label{it2:recons}  $\bigcup_{1\leq i \leq r} \{N_G(K_i)\}$ is a nested family
    \item\label{it3:recons}  $\forall u\in K_i, \forall v \in  N_G(K_i), \{u,v\} \in E$, $1\leq i \leq r$. In other words, $K_i$ is a module of $G$.
\end{enumerate}
\end{lemma}

In the remainder of this paper, a \emph{$k$-editing of $G$ 
into a trivially perfect graph} is a set $F \subseteq [V]^2$ such that $\size{F} \leq k$ and the 
graph $H = (V, E \triangle F)$ is trivially perfect.  
Here $E \triangle F = (E \cup F) \setminus (E \cap F)$ 
denotes the symmetric difference between sets $E$ and $F$. 
For the sake of readability, we simply speak of $k$-editing of $G$. 
We say that $F$ is a $k$-completion (resp. $k$-deletion) when $H = (V, E \cup F)$ 
(resp. $H = (V, E \setminus F)$) is trivially perfect. 
A vertex is \emph{affected} by a $k$-editing $F$ 
if it is contained in some pair of $F$ and 
\emph{unaffected} otherwise.  

\paragraph*{Packing, anti-matching and blow-up} 

We now define some structures and operators that will be useful for our kernelization algorithm. We assume in the remainder of this section that we are given a graph $G=(V,E)$. 
The notion of \emph{\PP{}} will be used in reduction rules to ensure the 
existence of unaffected vertices in ordered sets of critical cliques or of 
trivially perfect modules. 

\begin{definition}[\PP]
    Let $\mathcal{S} = \{C_1, \ldots, C_q\}$ be an ordered 
    collection of pairwise disjoint subsets of $V$.
    We say that $\mathcal{C} \subseteq \mathcal{S}$ is a \emph{\PP{}} of $\mathcal{S}$ 
    if $\mathcal{C} = \{C_1, \ldots, C_p\}$ for $1 \leqslant p \leqslant q$, 
    $\sum_{i=1}^p\size{C_i} \geqslant r$ and the number of vertices contained in $\mathcal{C}$ is minimum for this property. 
\end{definition}

\noindent In a slight abuse of notation we use $\mathcal{C}$ to denote both $\{C_1, \ldots, C_p\}$ and the set $\cup_{i=1}^p C_i$.

\begin{observation}
    \label{obs:packing}
    Let $\mathcal{S} = \{C_1, \ldots, C_q\}$ be an ordered 
    collection of pairwise disjoint subsets of $V$ such that $\size{C_j} \leqslant c$, for $1 \leqslant j \leqslant q$ and some integer $c > 0$. 
    Let $\mathcal{C} = \{C_1, \ldots, C_p\}$ be a \PP{} of $\mathcal{S}$. Then $\sum_{i=1}^p\size{C_i} \leqslant r+(c-1)$. 
\end{observation}

\begin{proof}
    Since $\sum_{i=1}^p\size{C_i} \geqslant r$ and 
    the number of vertices in $\mathcal{C}$ 
    is minimum for this property 
    we have that $\sum_{i=1}^{p-1}\size{C_i} \leqslant r-1$. The result follows from the fact that $\size{C_p} \leqslant c$. 
\end{proof}

\begin{definition}[Anti-matching]
    \label{def:am}
    An \emph{anti-matching} of $G$ is a set of pairwise disjoint pairs $\{u,v\}$ of vertices of $G$ such that $\{u,v\} \not \in E$.
\end{definition} 

\noindent In a slight abuse of notation we denote by $V(D)$ the set of vertices contained in pairs of an anti-matching $D$. 

\begin{observation}
    \label{obs:c4dam}
    Let $(G,k)$ be a \YINST{} of \TPE{} and $M$ be a module containing a $(k+1)$-sized anti-matching. 
    Let $F$ be a $k$-editing of $G$ and $H = G \triangle F$. 
    Then $N_G(M)$ is a clique in $H$.
\end{observation}

\begin{proof}
 Let $D = \{\{u_i,v_i\} \mid 1 \leqslant i \leqslant k+1\}$ be a $(k+1)$-sized anti-matching of $M$. Assume for a contradiction that $N_G(M)$ is not a clique in $H$ and let $\{u,v\}$ be a non-edge of $H$ with $u,v \in N_G(M)$. 
 Since $\size{F} \leqslant k$ there 
    exists $1 \leqslant j \leqslant k+1$ 
    such that $\{u_j,v_j\} \notin F$ and for every $x\in V(G)\setminus M$, $\{u_j,x\},\{v_j,x\}\notin F$. 
    Hence $\{u_j,u,v_j,v\}$ induces a $C_4$ in $H$, a contradiction.
\end{proof}

We conclude this section by introducing a gluing operation on trivially perfect graphs,  
namely \emph{blow-up}, that will ease the design of some reduction rules. 

\begin{definition}[Blow-up]
     \label{def:blowup}  
     Let $u$ be a vertex of $G = (V,E)$ and 
     $H = (V_H, E_H)$ be any graph. 
     The \emph{blow-up of $G$ by $H$ at $u$}, denoted 
     $\BLOW{G}{H}{u}$ is the graph 
     obtained by replacing 
     $u$ by $H$ in $G$. More formally: 
     $$
     \BLOW{G}{H}{u}= \big ( (V \setminus \{u\}) \cup V_H , E(G \setminus \{u\})\cup  E_H  \cup (V_H \times N_G(u) \big )
     $$
\end{definition}

\begin{proposition}
     \label{prop:blowup}
     Assume that $G$ is trivially perfect and let $u$ be a vertex of 
     $G$ such that $N_G[u]$ is a clique. For any trivially perfect 
     graph $H$, the graph $\BLOW{G}{H}{u}$ is trivially perfect. 
 \end{proposition}

\begin{proof}
    Let $S \subseteq V \setminus \{u\}$ be any maximal clique of $G$ 
    containing $N_G(u)$. We apply the forward direction of  \cref{lem:TP_recons} on $S$ to obtain components $\{K_1, \ldots, K_r\}$ that are modules such that 
    $G[S \cup K_i]$ is trivially perfect for every $1 \leqslant i \leqslant r$ 
    and $\bigcup_{1\leq i \leq r} \{N_G(K_i)\}$ is a nested family. Note that by 
    construction and w.l.o.g., we may assume $K_1 = \{u\}$. The result then 
    directly follows from the reverse direction of \cref{lem:TP_recons} by 
    replacing $K_1$ by $H$.
\end{proof} 

\section{Reduction rules}
\label{sec:rules}

In the remainder of this section we assume that we are given an instance $(G = (V,E), k)$ of \TPE{}.

\subsection{Standard reduction rules}
\label{subsec:modules}
We first describe some well-known reduction rules~\cite{BPP10,BP13,DP18,DPT23} 
that are essential to obtain a vertex-kernel using \cref{thm:strategy}~\cite{DPT23}. 
We will assume in the remainder of this work that the instance at hand is reduced under \cref{rule:compTP,rule:borneCC-TP}, 
meaning that none of them applies to the instance. 

\begin{polyrule}
    \label{rule:compTP}
   Let $C \subseteq V$ be a connected component of $G$ such that $G[C]$ is trivially perfect.
    Remove $C$ from $G$.
\end{polyrule}

\begin{polyrule} 
\label{rule:borneCC-TP}
    Let $K \subseteq V$ be a critical clique of $G$ 
    such that $\size{K} > k+1$. Remove $\size{K} -(k+1)$ arbitrary vertices in $K$ from $G$.
\end{polyrule}

\begin{lemma}[Folklore, \cite{BPP10,DP18}]
\label{lem:simple}
    \cref{rule:compTP,rule:borneCC-TP} are safe and can be applied in
    polynomial time.
\end{lemma}

\subsection{An \texorpdfstring{$O(k)$}{O(k)} bound on the size of trivially perfect modules} 
\label{subsec:bound}

Using an additional reduction rule bounding the size of independent sets in any trivially perfect module by $O(k)$, 
Drange and Pilipczuk~\cite{DP18} proved that such modules can be reduced to $O(k^2)$ vertices.
We strengthen this result by proving that trivially perfect modules can further be reduced to $O(k)$ vertices. We first deal with modules that contain a large anti-matching. 

\begin{polyrule}
    \label{rule:dam} 
    Let $M \subseteq V$ be a trivially perfect module of $G$. 
    If $G[M]$ contains a $(k+1)$-sized anti-matching $D$, then remove the vertices contained in $M \setminus V(D)$. 
\end{polyrule}

\begin{lemma}
    \label{lem:moduleTP}
    \cref{rule:dam} is safe. 
\end{lemma}

\begin{proof}
    Let $G' = (V',E')$ be the graph obtained after application of~\cref{rule:dam}. We need to prove that \INST{} is a \YINST{} if 
    and only if $(G'=(V',E'),k)$ is a \YINST. 
    The forward direction is straightforward 
    since $G'$ is an induced subgraph of $G$ and trivially perfect graphs are hereditary. 
    We now consider the reverse direction. Let $M' = V(D)$, the set of vertices kept by \cref{rule:dam}.  Moreover, let $F'$ be a $k$-editing of $G'$ and $H'=G' \triangle F'$. 
    We will construct a $k$-editing $\OPT{F}$ of $G$. 
    Note that since the pairs contained in an anti-matching are disjoint (\cref{def:am}), $\size{M'} = 2(k+1)$. 
    Moreover, since $\size{F'} \leqslant k$ there are at most $2k$ 
    affected vertices. 
    Hence let $u$ be an unaffected vertex of $M'$.  
    By \cref{obs:c4dam} and since $M'$ contains a $(k+1)$-sized anti-matching 
    we have that $N_{G'}(M')$ is a clique in $H'$. 
    The graph $H_u = H' \setminus (M' \setminus \{u\})$ is trivially perfect by heredity and $N_{H_u}(u) = N_{G'}(M')$.  
    It follows that $N_{H_u}(u)$ is a clique and \cref{prop:blowup} implies that the graph 
    $H = \BLOW{H_u}{M}{u}$ is trivially perfect. 
    Let $\OPT{F}$ be the editing such that $H = G \triangle \OPT{F}$.  Since $u$ is unaffected by $F'$ and $u \in M$ 
    we have $N_{H_u}(u) = N_G(M)$. Hence, since $M$ is 
    a module in $G$ we have that $N_{H}(v) = N_G(v)$ for every 
    vertex $v \in M$,  
    implying that $\OPT{F} \subseteq F'$. This concludes the proof. 
\end{proof}

In order to bound the size of any trivially perfect module by $O(k)$, we actually prove a more general reduction rule that will be useful for the rest of our kernelization algorithm. 
This rule operates on a more intricate structure, so-called \emph{comb}~\cite{DPT23}, that induces a trivially perfect graph but not necessarily a module. 

\begin{figure}[h]
    \centering
    \includegraphics[scale=1.1]{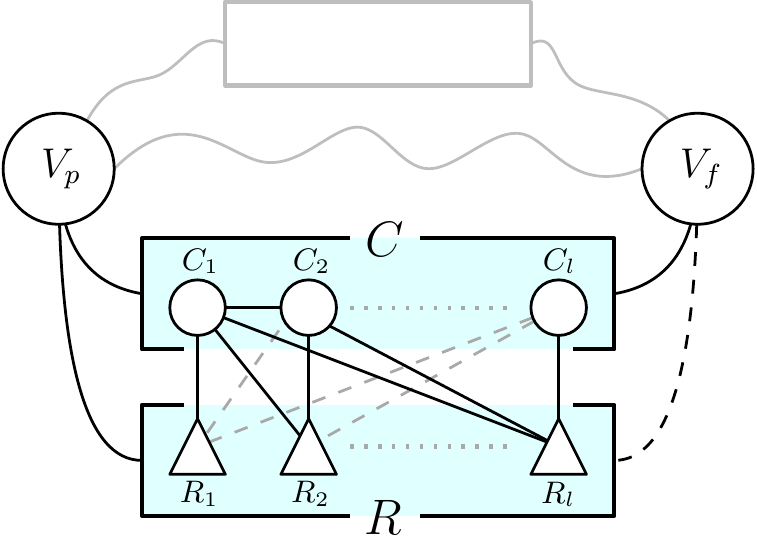}
    \caption{A comb of a graph $G = (V,E)$ with shaft $C$ and teeth $R$. Each set $C_i$ is a critical clique while each set $R_i$ induces a (possibly disconnected) trivially perfect module, $1 \leqslant i \leqslant l$. Notice that the sets $V_p$ and $V_f$ might be adjacent to some other vertices of the graph. \label{fig:comb}}
\end{figure}

\begin{definition}[Comb~\cite{DPT23}]\label{def:comb} 
A pair $(C,R)$ of disjoint subsets of $V$ is a comb of $G$ if:
\begin{itemize}
    \item $G[C]$ is a clique that can be partitioned into $l$
    critical cliques $\{C_1,...,C_l\}$
    \item $R$ can be partitioned into $l$ non-empty non-adjacent trivially perfect modules $\{R_1,...,R_l\}$
    \item $N_G(C_i) \cap R= \bigcup_{j=i}^l R_j$ and $N_G(R_i) \cap C= \bigcup_{j=1}^i C_j$ for $1 \leq i \leq l$
    \item there exist two (possibly empty) subsets of vertices $V_f,\VP \subseteq V(G)\backslash\{C\cup R\}$ such that: 
    \begin{itemize}
        \item $\forall x \in C,\ N_G(x) \backslash (C \cup R) = \VP \cup V_f$ and
        \item $\forall y\in R,\ N_G(y)\backslash (C \cup R) = \VP$. 
    \end{itemize}
\end{itemize}
\end{definition}

Given a comb $(C,R)$, $C$ is called the \emph{shaft} of the comb and $R$ the \emph{teeth} of the comb.  
See \cref{fig:comb} for an illustration of \cref{def:comb}. 
Recall that we assume that the graph is reduced under \cref{rule:borneCC-TP}, 
which means that $\size{C_i} \leqslant k+1$ for $1 \leqslant i \leqslant l$. 
\DPT{} showed the following proposition on the structure of combs. 

\begin{proposition}[\cite{DPT23}]\label{prop:unique}
Given a comb $(C,R)$ of $G$, the subgraph $G[C \cup R]$ is trivially perfect. Moreover the sets $V_p$ and $V_f$, and the ordered partitions $(C_1,\dots, C_l)$ of $C$ and $(R_1,\dots,R_l)$ of $R$ are uniquely determined.
\end{proposition}

In the following we assume that any comb $(C,R)$ is given with the 
ordered partitions $(C_1,\dots, C_l)$ of $C$ and $(R_1,\dots,R_l)$ of $R$. 
We note here that \cref{def:comb} slightly differs from the one given in~\cite{DPT23} where the set 
$V_f$ was required to be non-empty for technical reasons. Dropping this constraint will ease the presentation of our reduction rules. 

\smallskip

We now give several observations that will help understand \cref{def:comb}, in particular its relation to trivially perfect modules.  Given a trivially perfect graph $G = (V,E)$ and its UCD $\mathcal{T}_G = (T , \mathcal{B})$, one can construct a comb $(C,R)$ of $G$ by simply taking a path $P$ from a node $v_1$ of $T$ to one of its descendent $v_l$. The shaft $C$ are the vertices in bags of this path, the teeth $R$ are the bags of subtrees rooted in the children (not on $P$) of any node on the path $P$. We can observe that in this case, $V_p$ corresponds to vertices in the bags on the path from the parent of $v_1$ to the root of $T$ and that $V_f$ is empty. 

In particular, the vertex set of any connected trivially perfect graph can be partitioned into a comb $(C,R)$ by taking a path from the root of its UCD to one of its leaves. 
This means that when $V_p = V_f = \emptyset$, \cref{def:comb} corresponds to a 
connected trivially perfect graph. Similarly, if only the set $V_f$ is empty then 
\cref{def:comb} corresponds to a \emph{connected trivially perfect module} 
since for every $u \in C \cup R$ it holds that $N_G(u) \setminus (C \cup R) = V_p$. 

\smallskip

The following directly comes from the definition of a comb and is verified 
whether sets $V_p$ and $V_f$ are empty or not. 

\begin{observation}
    \label{obs:combmodules}
    The set of vertices $C$ (resp. $R$) is a module of $G \setminus R$ (resp. $G \setminus C$). 
\end{observation}
 
\smallskip

We will show that combs can be safely reduced to $O(k)$ vertices. We first focus on 
combs having a large shaft, which will allow us to 
reduce trivially perfect modules with small anti-matching to $O(k)$ vertices (\cref{lem:reducedmodule}).
Then we turn our attention to combs with many vertices in the teeth to bound the size of every comb to $O(k)$ vertices (\cref{lem:combs}). 

\paragraph*{Combs with large shafts}

\DPT{} showed that the length of a comb (\ie the number $l$ of different critical cliques in the shaft) can be reduced linearly in $k$. 
However, as critical cliques contain $O(k)$ vertices by \cref{rule:borneCC-TP}, it only allowed the authors to bound the number of vertices in shafts of combs to $O(k^2)$. 
\cref{rule:shaft} presented in this section keeps two sets $\mathcal{C}_a$ and $\mathcal{C}_b$ containing a linear number (in $k$) of vertices at the beginning and at the end of the shaft, allowing to bound its size linearly in $k$. 
The two sets $\mathcal{C}_a$ and $\mathcal{C}_b$ will be large enough to ensure the existence of two vertices that will be unaffected by a given $k$-editing of the graph.
We will use such vertices to prove that there exists a $k$-editing of the graph that does not affect any vertex in the shaft lying between $\mathcal{C}_a$ and $\mathcal{C}_b$, implying the safeness of the rule. 

\begin{polyrule}
    \label{rule:shaft} 
    Let $(C,R)$ be a \WC{} of $G$ such that there exist \emph{disjoint} \SPP{2k+1}s 
    $\mathcal{C}_a$ of $\{C_1, \ldots, C_l\}$ and $\mathcal{C}_b$ of $\{C_l, C_{l-1}, \ldots, C_1\}$.  
    Remove $C' = C \setminus (\mathcal{C}_a \cup \mathcal{C}_b)$ from $G$.  
\end{polyrule}

\begin{lemma}
    \label{lem:shaft}
    \cref{rule:shaft} is safe. 
\end{lemma}

\begin{proof} 
    Let $G' = G \setminus C'$ be the graph obtained after application of~\cref{rule:shaft}. 
    Since $G'$ is an 
    induced subgraph of $G$ and since trivially perfect graphs are 
    hereditary, any $k$-editing of $G$ is a $k$-editing of $G'$. 
    
    For the reverse direction, let $F'$ be a $k$-editing of $G'$ and 
    $H'=G' \triangle F'$. We will construct a $k$-editing $\OPT{F}$ of $G$.  
    Let $c_a$ and $c_b$ be unaffected vertices 
    in $\mathcal{C}_a$ and $\mathcal{C}_b$, respectively. 
    Note that both sets contain at least $2k+1$ vertices and that 
    $F'$ affects at most $2k$ vertices, hence $c_a$ 
    and $c_b$ are well-defined. Let $C_a$ and $C_b$ be the 
    critical cliques of $C$ containing $c_a$ and $c_b$, $1 \leqslant a < b \leqslant l$. Moreover, let 
    $\CM = C_{a+1} \cup \ldots \cup C_{b-1}$ and   
    $\RM = R_a \cup \ldots \cup R_{b-1}$.  
    Similarly, let 
    $\CL = C_1 \cup \ldots \cup C_{a}$,  
    $\CU = C_b \cup \ldots \cup C_l$ and 
    $\RU = R_b \cup \ldots \cup R_l$
    These sets are depicted Figure~\ref{fig:shaft}. 
    Finally, let $\GM = G \setminus \CM$ and $\HM = H'\setminus \CM$. Notice in particular that $\HM$ is trivially perfect and that $C' \subseteq \CM$. 

    \begin{figure}[ht]
        \centering
        \includegraphics[scale=1.1]{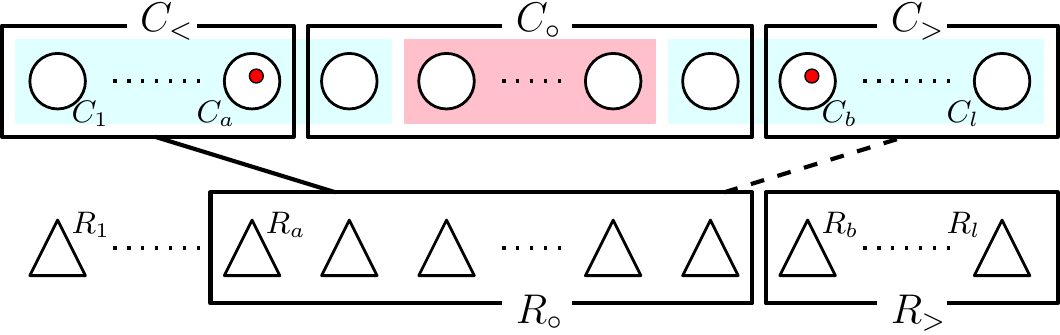}
        \caption{
        Illustration of the comb and the sets used in the proof of \cref{lem:shaft}. The circles are critical cliques of the shaft and the triangles are teeth. The red vertices correspond to $c_a$ and $c_b$, the light blue rectangles correspond to sets $\mathcal{C}_a$ and $\mathcal{C}_b$ and the light red rectangle corresponds to $C'$, which is removed by \cref{rule:shaft}. \label{fig:shaft}}
    \end{figure}
    
    Let $\FM \subseteq F'$ be the $k$-editing such that $\HM = \GM \triangle \FM$ and $\SM$ be a maximal clique of $\HM$ containing $\{c_a,c_b\}$. 
    Notice that since $c_a$ and $c_b$ are unaffected, 
    $\SM$ is included in $N_{\GM}(\{c_a,c_b\}) = C \cup V_p \cup V_f \cup \RU$. 
    We use \cref{lem:TP_recons} on $\SM$ 
    to obtain a set of connected components $\{K_1, \ldots, K_r\}$ 
    of $\HM\setminus \SM$ such that $\{K_1, \ldots, K_r\}$ are modules in 
    $\HM$ whose (possibly empty) neighborhoods in $\SM$ form a nested family. 
    We first modify $\FM$ to 
    obtain a $k$-editing of $\GM$ where vertices of $\RM$ are affected uniformly.  

    \begin{claim}
        \label{claim:RMmodule}
        There exists a $k$-editing $\OPT{F}$ of $\GM$ such that, 
        in $\OPT{H} = \GM \triangle \OPT{F}$, $\RM$ is a 
        module and $\OPT{H}[\RM] = G[\RM]$.
    \end{claim}

    \begin{proofclaim}
        We begin with several useful observations. 
        First, $\RM$ is a module in $\GM$ since $R \supset \RM$  
        is a module in $G \setminus C$ (\cref{obs:combmodules}) 
        and vertices of $\RM$ are adjacent to $\CL$ and non adjacent 
        to $\CU$. 
        Next, since any component $K_i$ is a module in $\HM$, $1 \leqslant i \leqslant r$, 
        and since $c_a$ and $c_b$ are unaffected by $\FM$, we have 
        $N_{\HM}(K_i) \cap \{c_a,c_b\} = N_{\GM}(K_i) \cap \{c_a,c_b\}$. 
        In other words, vertices in a same component $K_i$ must have the same adjacency  
        with $\{c_a,c_b\}$ in $\GM$ and in $\HM$. 
        Similarly, no vertex $v \in \RM$ belongs to $\SM$ since 
        $N_{\GM}(v) \cap \{c_b\} = \emptyset$. 
        Moreover, the only vertices of $\GM$ that are adjacent to 
        $c_a$ but not $c_b$ are exactly those of $\RM$. Hence for any 
        vertex $\VM \in \RM$ it holds that $N_{\HM}(\VM) \subseteq \SM \cup \RM$. 

        \smallskip

        Assume now that $\RM$ is not a module in $\HM$ and let 
        $\VM \in \RM$ be a vertex contained in the least number of pairs 
        of $\FM$ with the other element in $\SM$. Consider the graph 
        $\tilde{H} = \HM \setminus (\RM \setminus \{\VM\})$, which is trivially perfect by heredity. Since $N_{\HM}(\VM) \subseteq \SM \cup \RM$, it follows that $N_{\tilde{H}}(\VM) \subseteq \SM$ is a clique. Hence \cref{prop:blowup} implies that 
        the graph $\OPT{H} =\BLOW{\tilde{H}}{G[\RM]}{\VM}$ is trivially perfect. 
        Let $\OPT{F}$ be the editing such that $\OPT{H} = \GM \triangle \OPT{F}$. 
        By the choice of $\VM$ 
        we have $\size{\OPT{F}} \leqslant \size{\FM}$. 
        It follows that $\OPT{F}$ is a desired $k$-editing, concluding the proof of \cref{claim:RMmodule}. 
    \end{proofclaim}

    We henceforth consider $\OPT{H} = \GM \triangle \OPT{F}$ 
    where $\OPT{F}$ is the $k$-editing from  \cref{claim:RMmodule}. 
    Note that the components around $\SM$ may be different in $\HM \setminus \SM$ and $\OPT{H} \setminus \SM$. 
    In a slight abuse of notation, we still define these components by 
    $\{K_1, \ldots, K_r\}$. Recall that $\{K_1, \ldots, K_r\}$ are modules in 
    $\OPT{H}$ whose (possibly empty) neighborhoods in $\SM$ form a nested family. 

    \begin{claim}
        \label{claim:tp}
        The graph $H = G \triangle \OPT{F}$ is trivially perfect.
    \end{claim}

    \begin{proofclaim}
        The graph $H$ corresponds to $\OPT{H}$ where vertices of 
        $\CM$ have been added with the same neighborhood as in $G$. Let us first observe 
        that $S = \SM \cup \CM$ is a maximal clique in $H$. 
        Indeed, $\CM$ is a clique by definition and $\SM \subseteq \big ( C \cup V_p \cup V_f \cup \RU \big ) \subseteq N_{H}(\CM) = N_G(\CM)$
        (recall that $C$ is adjacent to $V_p \cup V_f$ by \cref{def:comb} and 
        that vertices of $\CM$ are adjacent to every vertex of $\RU$). 
        Hence components $\{K_1, \ldots, K_r\}$ defined in $\OPT{H} \setminus \SM$ are the same in $H \setminus S$ and 
        their neighborhoods are nested in $\SM$. 
        We split $\{K_1, \ldots, K_r\}$ into three 
        types components w.r.t their adjacencies with  $\{c_a,c_b\}$, namely: 
    
        \begin{enumerate}
            \item $\alpha$-components that are non-adjacent to both $c_a$ and $c_b$
            \item $\beta$-components that are adjacent to $c_a$ but not $c_b$
            \item $\delta$-components that are adjacent to both $c_a$ and $c_b$
        \end{enumerate}

        In what follows we let $K_\alpha$, $K_\beta$ and 
        $K_\delta$ denote any $\alpha$-, $\beta$- and $\delta$-component, respectively. Note that 
        $N_{\OPT{H}}(K_\alpha) \subseteq N_{\OPT{H}}(K_\beta) \subseteq N_{\OPT{H}}(K_\delta) \subseteq \SM$ holds by construction. 
        Recall that since $c_a$ and $c_b$ are unaffected by $\OPT{F}$,
        $N_{G}(K_i) \cap \{c_a,c_b\} = N_{H}(K_i) \cap \{c_a,c_b\}$ for any $K_i$.  
        We claim that $\{N_{H}(K_i) \mid 1 \leqslant i \leqslant r\}$ is a nested family. Note that \cref{lem:TP_recons} will imply the result since 
        $S$ is a maximal clique in $H$. 
        To sustain this claim, recall that the neighborhoods of vertices of $\CM$ are identical in $G$ and $H$.
        Moreover, $N_H[c_b] \subseteq N_H[\CM]\subseteq N_H[c_a]$ holds as these vertices are unaffected by $\OPT{F}$. It follows that $\alpha$-components (resp. $\delta$-components) are non-adjacent (resp. adjacent) to every vertex of $\CM$ in $H$. 
        This means in particular that the neighborhoods of both $\alpha$- and $\delta$-components are nested in $S$.
        Moreover we can observe that vertices of $\beta$-components are exactly the ones of $\RM$ since they are the only ones that are adjacent to $c_a$ but not $c_b$ in $G$.  
        Hence, in $H$, we still have:
        \[
            N_{H}(K_\alpha) \subseteq N_{H}(K_\beta) \subseteq N_{H}(K_\delta) 
        \]
        
        It remains to prove that the neighborhoods of  $\beta$-components are nested in $S$.
        Let w.l.o.g. $\{K_1, \ldots, K_p\}$, $1 \leqslant p \leqslant r$ 
        be the $\beta$-components. 
        By definition of a comb, the $\beta$-components (which are also $\RM$) can be ordered w.r.t. the inclusion of their neighborhood in $G[\CM]$. We can assume w.l.o.g. that the ordering is $N_{G[\CM]}(K_1) \subseteq \ldots \subseteq N_{G[\CM]}(K_p)$. 
        Moreover we can observe that for any $\beta$-component $K_i$ we have $N_{G[\CM]}(K_i) = N_{H[\CM]}(K_i)$, $1 \leqslant i \leqslant p$. 
        Since $\RM$ is a module in $\OPT{H}$ by \cref{claim:RMmodule}  and since vertices of $\beta$-components 
        are exactly those of $\RM$,
        it follows that the neighborhoods of $\beta$-components are nested. 
        Hence $\{N_{H}(K_i) \mid 1 \leqslant i \leqslant r\}$ is a nested family and $H$ is a trivially perfect graph by \cref{lem:TP_recons}. 
    \end{proofclaim}
    
    By ~\cref{claim:tp} the graph $H = G \triangle \OPT{F}$ is trivially perfect and as $|\OPT{F}| \leq k$, it follows that $\OPT{F}$ is a $k$-editing of $G$, concluding the proof of \cref{lem:shaft}.
\end{proof}

\begin{observation}
    \label{obs:shaft}
    Assume that the instance $(G,k)$ is reduced under 
    \cref{rule:borneCC-TP,rule:shaft}. For any comb 
    $(C,R)$ of $G$ it holds that $\size{C} \leqslant 6k+2$. 
\end{observation}

\begin{proof}
    Since $G$ 
    is reduced under \cref{rule:borneCC-TP} every critical clique 
    $C_i$ of the shaft contains at most $k+1$ vertices, $1 \leqslant i \leqslant l$. 
    By \cref{obs:packing}, any \SPP{2k+1} of $\{C_1, \ldots, C_l\}$ (resp. $\{C_l, \ldots, C_1\}$) contains at most 
    $3k+1$ vertices. It follows that $\size{C} \leqslant 6k+2$ since 
    otherwise one could find two \emph{disjoint} \SPP{2k+1}s of $\{C_1, \ldots, C_l\}$ and of $\{C_l,  \ldots, C_1\}$ 
    and \cref{rule:shaft} would apply. 
\end{proof}

We are now ready to show how to reduce 
the size of any trivially perfect module. We need a combinatorial result 
that will be useful to obtain the claimed bound. 

\begin{lemma}
    \label{lem:nodam} 
    Let $G = (V,E)$ be a connected trivially perfect graph 
    and $\alpha$ be the size 
    of a maximum anti-matching of $G$. 
    There exists a \WC{} $(C,R)$ of $G$ such that $V = C\cup R $ and  
    $\size{R} \leqslant 4 \alpha$. 
    Moreover, such a comb can be computed in polynomial time. 
\end{lemma}

\begin{proof}
    We provide a constructive proof that will directly imply the last part of the result. 
    Recall that any trivially perfect graph contains 
    a universal vertex and let $U_1 \subseteq V(G)$ be the universal clique of $G$. 
    Let $R^1_1, \ldots, R^1_{p_1}$ denote the connected components of $G \setminus U_1$. 
    Since $G$ does not contain any $(\alpha+1)$-sized anti-matching, there is at most one set $R^1_i$, 
    $1 \leqslant i \leqslant p_1$ such that $\size{R^1_i} > \alpha$ (as there is no edge between $R^1_i$ and $R^1_j$, 
    $1 \leqslant i < j \leqslant p_1$).  

    Assume without loss of generality that $\size{R^1_1} > \alpha$. 
    We add all vertices of $\cup_{i=2}^{p_1}R^1_i$ to some set $\RL$ and we will repeat this process on $G[R^1_1]$ until every connected component is smaller than $\alpha$. 
    More formally, at step $j>1$, for the trivially perfect graph $G_j = G[R^{j-1}_1]$, let $U_j$ be its universal clique and $R^j_1, \dots, R^j_{p_j}$ be the connected components of $G_j\setminus U_j$. Let $ R^j_1$ be the one of size greater than $\alpha$ if it exists, if it does not, stop the process and let $l$ be the last step. In particular, $|R^l_i|\leq \alpha$, $1\leq i\leq p_l$.
    Let $\RL = \bigcup_{j=1}^{l-1}\bigcup_{i=2}^{p_j} R_i^j $ and $\RU = R^l_1 \cup \dots \cup  R^l_{p_l}$. 

    Recall that $\size{R^{l-1}_1} > \alpha$ by construction. 
    This implies that $\size{\RL} \leqslant \alpha$ since otherwise $G[\RL \cup R^{l-1}_1]$ would contain a 
    $(\alpha+1)$-sized anti-matching. 
    We claim that $\size{\RU} \leqslant 3\alpha$. 
    To support this claim, let us consider the \SPP{\alpha+1} $\{R^l_1, \ldots, R^l_q\}$ of $\{R^l_1 ,\ldots, R^l_{p_l}\}$ and let $R' = \bigcup_{i=1}^q R_i^l$ be its vertices. Let $R'' = \RU \setminus R'$. Recall that $l$ is the last step of the process and $\size{R_i^l} \leqslant \alpha$ for $1 \leqslant i \leqslant p_l$. 
    Hence by \cref{obs:packing} it holds that $|R'| \leqslant 2\alpha$.  
    Thus, we have that $\size{R''} \leqslant \alpha$ since otherwise $G[R' \cup R'']$ would contain 
    a $(\alpha+1)$-sized anti-matching, 
    a contradiction. Hence $\size{\RU} = \size{R'} + \size{R''} \leqslant 3\alpha$. 
    
    To obtain a \WC{} for $G$ we consider the set $C = \{U_1, \ldots, U_l\}$ as the shaft 
    (recall that $U_1$ is the universal clique of $G$ and that $U_j$ denotes the universal clique of $G[R^{j-1}_1]$ at every 
    step $1 < j \leqslant l$). 
    Moreover, for every $1 \leqslant j < l$, the tooth $R_j$ is equal to $R_2^j\cup \ldots\cup R_{p_j}^j$, 
    the last tooth $R_l$ being $R_>$. By construction $(C, R = \bigcup_{j=1}^l R_j)$ is
    a \WC{} of $G$ such that $\size{R} = \size{R_<} + \size{R_>} \leqslant 4\alpha$.  
    This concludes the proof. 
\end{proof}

\begin{lemma}
    \label{lem:reducedmodule}
    Assume that the instance $(G,k)$ is reduced under  \cref{rule:compTP,rule:borneCC-TP,rule:dam,rule:shaft} 
    and let $M$ be a trivially perfect module of $G$. Then $M$ contains at most $11k+2$ vertices.  
\end{lemma}

\begin{proof}   
    Observe that if $M$ contains an anti-matching of size more than $k$, then it is reduced under \cref{rule:dam} and contains $2k+2$ vertices. 
    Hence, suppose that $M$ does not contain a $(k+1)$-sized anti-matching. 
    Assume first that $G[M]$ is connected. Let $(C,R)$ be a comb obtained through \cref{lem:nodam}, such that $C\cup R = M$ and $|R|\leq 4k$. 
    By \cref{obs:shaft} we have that $\size{C} \leqslant 6k+2$. 
    It follows that $\size{M} \leqslant \size{C} + \size{R} \leqslant 10k+2$. 
    
    \smallskip
    
    To conclude it remains to deal with 
    the case where $G[M]$ is disconnected. Let $\{M_1, \ldots, M_p\}$ denote the connected components of $G[M]$.
    As $M$ does not contain a $(k+1)$-sized anti-matching, at most one of its connected component has size greater than $k$, we may assume w.l.o.g. that it is $M_1$, if existent.
    Let $\mathcal{C}$ be the \SPP{k+1} of $\{M_1, \ldots, M_p\}$. As $|M_1|\leq 10k+2$ and $|M_i|\leq k$ for $2\leq i\leq p$, we have that $|\mathcal{C}|\leq 10k+2$.
    Moreover, since $M$ does not contain any $(k+1)$-sized anti-matching, $|M\setminus \mathcal{C}|\leq k$ and thus $|M|\leq 11k+2$.
    This concludes the proof.
\end{proof}

\subsection{Combs with large teeth} 
\label{subsec:teeth}

We now turn our attention to the case were a given comb  contains many vertices in its teeth. 
The arguments are somewhat symmetric to the ones used in the proof of \cref{lem:shaft}. 
The main difference lies 
in the fact that the information provided by unaffected 
vertices differ when they are contained in the teeth rather than in the shaft. 

\begin{polyrule}
    \label{rule:teeth}
    Let $(C,R)$ be a \WC{} of $G$ such that there exist three disjoint sets $\mathcal{R}_a, \mathcal{R}_b$ and $\mathcal{R}_c$ where:
    \begin{itemize}
        \item $\mathcal{R}_a$ is a \SPP{2k+1} of $\{R_1, \ldots, R_l\}$,
        \item $\mathcal{R}_c = \{R_l, \dots, R_q\}$ is a \SPP{2k+1} of $\{R_l, \ldots, R_1\}$,
        \item $\mathcal{R}_b$ is a \SPP{2k+1} of $\{R_{q-1}, \ldots, R_1\}$,
    \end{itemize}
    Remove $R' = R \setminus (\mathcal{R}_a \cup \mathcal{R}_b \cup \mathcal{R}_c)$ from $G$.  
\end{polyrule}

\begin{lemma}
    \label{lem:teeth}
    \cref{rule:teeth} is safe. 
\end{lemma}

\begin{proof}
    Let $G' = G \setminus R'$ be the graph obtained after application of~\cref{rule:teeth}. 
    Since $G'$ is an 
    induced subgraph of $G$ and since trivially perfect graphs are 
    hereditary, any $k$-editing of $G$ is a $k$-editing of $G'$. 
    
    For the reverse direction, let $F'$ be a $k$-editing of $G'$ and $H'=G' \triangle F'$. We will construct a $k$-editing $\OPT{F}$ of $G$.   
    Let $r_a$, $r_b$ and $r_c$ be unaffected vertices 
    in $\mathcal{R}_a$, $\mathcal{R}_b$ and $\mathcal{R}_c$, respectively.
    Note that these vertices exist as these sets contain at least $2k+1$ vertices and $F'$ affects at most $2k$ vertices. 
    Let $R_a$, $R_b$ and $R_c$, $1 \leqslant a < b < c  \leqslant l$,  
    be the 
    teeth of $R$ containing $r_a$, $r_b$ and $r_b$, respectively (these sets are well-defined since the packings $\mathcal{R}_a$, $\mathcal{R}_b$ and $\mathcal{R}_c$ are disjoint). 
    Moreover, since $r_a$, $r_b$ and $r_c$ are unaffected by 
    $F'$ their neighborhoods are equal in  
    $G'$ and $H'$ and hence $\big ( N_{H'}(r_a) \setminus R_a \big ) \subseteq \big ( N_{H'}(r_b) \setminus R_b \big ) \subseteq \big (N_{H'}(r_c) \setminus R_c \big )$.

    \begin{claim}
        \label{claim:rbclique}
        The set $N_{H'}(r_b) \setminus R_b$ is a clique in $H'$. 
    \end{claim}

    \begin{proofclaim}
        Assume for a contradiction that $N_{H'}(r_b) \setminus R_b$ 
        contains a non-edge $\{u,v\}$. 
        Recall that there is no edge between $R_b$ and $R_c$. Hence, 
        since $\big ( N_{H'}(r_b) \setminus R_b \big ) \subseteq \big ( N_{H'}(r_c) \setminus R_c \big )$ 
        we have that the set $\{r_b,u,v,r_c\}$ induces a $C_4$ in $H'$, 
        a contradiction. 
    \end{proofclaim}    
    
    Let $\RM = R_{a+1} \cup \ldots \cup R_{b-1}$ 
    and $\CM= C_{a+1} \cup \ldots \cup C_{b}$.   
    Similarly, let 
    $\CL = C_1 \cup \ldots \cup C_a$,   
    $\RL = R_1 \cup \ldots \cup R_{a}$  
    and $\RU = R_b \cup \ldots \cup R_l$.   
    Finally, let $\GM = G \setminus \RM$ and $\HM = H'\setminus \RM$. 
    These sets are depicted Figure~\ref{fig:teeth}.   
    Notice in particular that 
    $\HM$ is trivially perfect and that 
    $R' \subseteq \RM$. 
    Let $\FM \subseteq F'$ be the $k$-editing such that $\HM = \GM \triangle \FM$.  
    We first modify $\FM$ to 
    obtain a $k$-editing of $\GM$ where every vertex of $\CM$ is affected uniformly.  

    \begin{figure}[ht]
        \centering
        \includegraphics[scale=1.1]{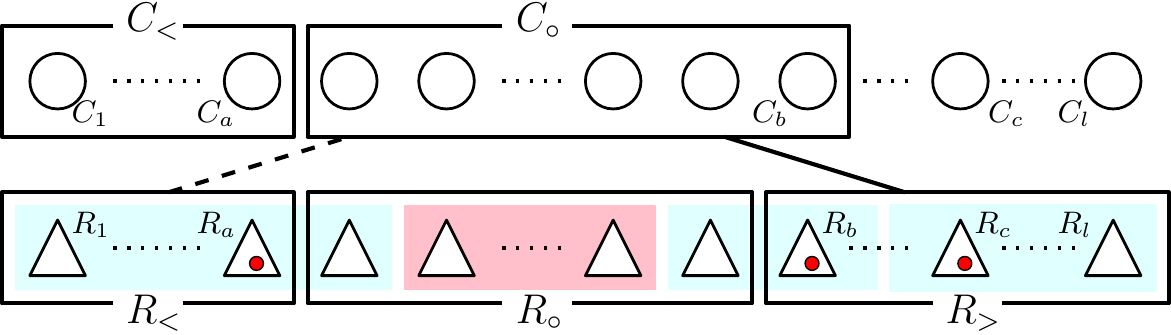}
        \caption{Illustration of the comb and the sets used in the proof of \cref{lem:teeth}. The circles are critical cliques of the shaft and the triangles are teeth. The red vertices correspond to $r_a$, $r_b$ and $r_c$, the light blue rectangles correspond to sets $\mathcal{R}_a$, $\mathcal{R}_b$ and $\mathcal{R}_c$ and the light red rectangle corresponds to $R'$, which is removed by \cref{rule:teeth}. \label{fig:teeth}}
    \end{figure}
    
    \begin{claim}
        \label{claim:CMmodule}
        There exists a $k$-editing $\OPT{F}$ of $\GM$ such that, 
        in $\OPT{H} = \GM \triangle \OPT{F}$, $\CM$ is a 
        clique module.
    \end{claim} 

    \begin{proofclaim}
        Note that $\CM$ is a critical clique in $\GM$ since 
        $C \supset \CM$ is a module in $G \setminus R$ (\cref{obs:combmodules}) and vertices of $\CM$ are non-adjacent to vertices of $\RL$ and adjacent to vertices of $\RU$. 
        \smallskip
        Assume now that $\CM$ is not a clique module in $\HM$ and let 
        $\VM \in \CM$ be a vertex contained in the least number of pairs 
        of $\FM$. 
        Consider the graph $\HM'  = \HM \setminus (\CM \setminus \{\VM\})$, which is trivially perfect by heredity, and 
        let $\OPT{H}$ be the graph obtained from $\HM'$
        by adding vertices of $\CM \setminus \{\VM\}$ as true twins of $\VM$. Let $\OPT{F}$ be the editing such that $\OPT{H} = \GM\triangle \OPT{F}$. The graph $\OPT{H}$ is trivially perfect as the class of trivially perfect graphs is closed under true twin addition. 
        It follows from construction that $\CM$ is a clique module in $\OPT{H}$ and by the choice of $\VM$, $\size{\OPT{F}} \leqslant \size{\FM}$. 
    \end{proofclaim}

    We henceforth consider $\OPT{H} = \GM \triangle \OPT{F}$ 
    where $\OPT{F}$ is the editing from   \cref{claim:CMmodule}. 
    We now show that vertices of $\RM$ can be added into $\OPT{H}$ while ensuring it remains trivially perfect. 

    \begin{claim}
        \label{claim:tpteeth}
        The graph $H = G \triangle \OPT{F}$ is trivially perfect.
    \end{claim}

    \begin{proofclaim}
        We start by removing the vertices of $R_b \setminus \{r_b\}$ from $\OPT{H}$, which will give us more control on the neighborhood
        of $r_b$ and ease some arguments. 
        Let $\tilde{H} = \OPT{H} \setminus (R_b\setminus \{r_b\})$, this graph is trivially perfect by heredity. 
        Let $S$ be a maximal clique of $\tilde{H}$ containing $r_b$. 
        By \cref{claim:rbclique}, $ N_{\tilde{H}}(r_b)$ is a clique and since $r_b$ is 
        unaffected by $\OPT{F}$ we have that
        $S = N_{\tilde{H}}[r_b] = \CL \cup \CM \cup V_p\cup \{r_b\}$.     
        We use \cref{lem:TP_recons} on $S$ 
        to obtain a set of connected components $\{K_1, \ldots, K_r\}$ 
        of $\tilde{H} \setminus S$ such that $\{K_1, \ldots, K_r\}$ are modules in 
        $\tilde{H}$ whose (possibly empty) neighborhoods in $S$ form a nested family. 
        We further split components 
        $\{K_1, \ldots, K_r\}$ into two types: $K_i$ is an $\alpha$-component if $N_{\tilde{H}}(K_i) \subseteq \big ( N_{\tilde{H}}(r_a) \cap S \big )$ and a $\beta$-component otherwise, $1 \leqslant i \leqslant r$. Since $N_{\tilde{H}}(r_a) \cap S = V_p \cup \CL$ we have that, for any $\alpha$-component $K_\alpha$, $N_{\tilde{H}}(K_\alpha) \subseteq V_p \cup \CL$. 
        Moreover, since $S = N_{\tilde{H}}[r_b]$ and since $\CM$ is a clique module in $\tilde{H}$ by \cref{claim:CMmodule}, every $\beta$-component $K_\beta$ satisfies $N_{\tilde{H}}(K_\beta) = V_p\cup \CL \cup \CM =  S \setminus \{r_b\}$. 

        \smallskip

        Observe now that $( V_p \cup \CL ) \subseteq N_{G}(\RM) \subseteq S\setminus \{r_b\}$. 
        In other words, the neighborhood of any tooth of $\RM$ contains the neighborhood of any $\alpha$-component and is contained in the neighborhood of any $\beta$-component. 
        Moreover the neighborhoods of the teeth of $\RM$ are nested in $G$ by definition of a comb.
        It follows that the vertices of $\RM$ can be safely added to $\tilde{H}$ with the same neighborhood as they have in $G$, ensuring that the resulting graph $H_b$ is trivially perfect. 
        It remains to add the vertices of $R_b$ back into the graph. By \cref{claim:rbclique,prop:blowup}, the graph $H = \BLOW{H_b}{G[R_b]}{r_b}$ is trivially perfect. 
    \end{proofclaim}

    By \cref{claim:tpteeth} the graph $H = G \triangle \OPT{F}$ is trivially perfect and as $|\OPT{F}| \leq k$, it follows that $\OPT{F}$ is a $k$-editing of $G$, concluding the proof of \cref{lem:teeth}.
\end{proof}

\begin{lemma}
    \label{lem:combs}
    Assume that the instance $(G,k)$ is reduced under \cref{rule:compTP,rule:borneCC-TP,rule:dam,rule:shaft,rule:teeth}. Let 
    $(C,R)$ be a comb of $G$. Then $\size{C \cup R} = O(k)$. 
\end{lemma}

\begin{proof} 
    First, note that $|C| \leq 6k+2$ thanks to \cref{obs:shaft}. 
    We proceed in the same fashion to bound the size of $R$. 
    As the teeth of a comb are trivially perfect modules, \cref{lem:reducedmodule} implies that $\size{R_i} \leqslant 11k+2$, $1 \leqslant i \leqslant l$. 
    Hence by \cref{obs:packing} any \SPP{2k+1} of $\{R_1, \ldots, R_l\}$ requires at most 
    $13k+2$ vertices. It follows that $\size{R} \leqslant 39k+6$ 
    since otherwise one could find three disjoint \SPP{2k+1}s of $R$ 
    that meet the requirements of \cref{rule:teeth}. Altogether we obtain that $\size{C \cup R} \leqslant 45k+8$ which concludes the proof. 
\end{proof}

\subsection{Reducing the graph exhaustively} 
\label{subsec:reducing}

We conclude this section by showing that the graph can be reduced in polynomial time. 

\begin{lemma}
    \label{lem:rulespoly}
    There is a polynomial time algorithm that outputs an instance 
    $G' = (V',E')$ such that none of Rules \ref{rule:compTP} to \ref{rule:teeth} applies. 
\end{lemma}

\begin{proof}
    First, \cref{rule:compTP,rule:borneCC-TP} can be applied in polynomial time 
    thanks to \cref{lem:simple}. 
    We now need to apply the other rules on trivially perfect modules and combs. For the modules, it is sufficient to reduce \emph{strong modules}, which are modules that do not overlap with other modules. 
    We can enumerate strong modules in linear time~\cite{TDH+08}. For each strong module $M$ we can check in polynomial time if it is trivially perfect. 
    We can moreover check if $M$ contains a $(k+1)$-sized anti-matching by finding a maximum matching in the complement graph $\overline{G[M]}$, for instance using Edmonds' algorithm~\cite{Edmonds65}. 
    If $M$ has a large anti-matching, then we can apply \cref{rule:dam}.
    Otherwise, if $\size{M} \geq 11k+2$ then   
    it can be reduced using \cref{rule:shaft}. 
    Indeed, $G[M]$ contains in this case 
    at most one connected component $M'$ with more than $k$ 
    vertices, such that $\size{M \setminus M'} \leq k$ (since otherwise $M$ would contain a $(k+1)$-sized anti-matching).   
    We compute a comb $(C,R)$ through \cref{lem:nodam} in 
    $G[M']$, with $\size{R} \leqslant 4k$. It follows that 
    $\size{C}>6k+2$ and \cref{obs:shaft} implies that 
    \cref{rule:dam} applies. 

    It remains to show that the combs not included in a trivially perfect module can be reduced in polynomial time. In order to do this \DPT{} showed that so-called \emph{critical combs} can be enumerated in polynomial 
    time, a critical comb being an inclusion-wise maximal comb where $V_f \neq \emptyset$ and  
    $R \cup C \cup V_f$ does not induce a trivially perfect module.
    In particular, critical combs contain every comb not included in a trivially perfect module. Hence it is sufficient to only reduce these combs. Given a critical comb, \cref{rule:shaft,rule:teeth} can be applied in polynomial time. This concludes the proof. 
\end{proof}

Combining \cref{thm:strategy} and \cref{lem:reducedmodule,lem:combs,lem:rulespoly} we obtain the main result of this work. 

\begin{theorem}
    \label{thm:result} 
    \TPE{} admits a kernel with $O(k^2)$ vertices. 
\end{theorem}

\begin{proof}
    We give a formal proof of \cref{thm:strategy} for the sake of completeness. Note that most arguments and notations are extracted from the proof of~\cite[Theorem 1]{DPT23}. 
    Recall that $c(k)$ and $p(k)$ are functions defined as, respectively, the maximum size of a trivially perfect module in and a comb of $G$ in \cref{thm:strategy}. 
    Let $(G=(V,E),k)$ be a reduced yes-instance of \TPE{} and $F$ a $k$-editing set of $G$. 
    Let $H=G\triangle F$ and $\mathcal{T} = (T, \mathcal{B})$ the universal clique decomposition of $H$. 
    The graph $G$ is not necessarily connected, thus $T$ is a forest. Let $A$ be the set of nodes $t\in V(T)$ such that the bag $B_t$ contains a vertex affected by $F$. 
    Since $ \vert F \vert  \leq k$, we have $ \vert A \vert \leq 2k$. Let $A' \subseteq V(T)$ be the least common ancestor closure of $A$ plus the root of each connected component of $T$.
    The least common ancestor closure is obtained as follows: start with $A' = A$ and while there is $u,v\in A'$ whose least common ancestor $w$ (in $T$) is not in $A'$, add $w$ to $A'$. 
    According to~\cite[Lemma 1]{FLM+12} the least common ancestor closure of $A$ is of size at most $2\vert A\vert$. 
    Moreover, \cref{rule:compTP} implies that there are at most $2k$ connected components in $H$ and thus $2k$ roots, hence $ \vert A' \vert  \leq 6k$. 
    
    Let $D$ be a connected component of $T \setminus A'$. We can observe that, by construction of $A'$, only three cases are possibles:
    \begin{itemize}
        \item $N_T(D) = \emptyset$ ($D$ is a connected component of $T$).
        \item $N_T(D) = \{a\}$ ($D$ is a subtree of $T$ whose parent is $a \in A'$). 
        \item $N_T(D) = \{a_1,a_2\}$ with one of the nodes $a_1,a_2 \in A'$ being an ancestor of the other in $T$.
    \end{itemize}
    Dumas \etal~\cite{DPT23} denote these connected components as respectively of type $0$, $1$ or $2$. For $D \subseteq V(T)$, let $W(D) = \bigcup_{t\in D}B_t$ denote the set of vertices of $G$ corresponding to bags of $D$. 
    
    There is no connected component of type 0 or else $W(D)$ would be a connected component of $G$ inducing a trivially perfect graph. 
    Rule~\ref{rule:compTP} would have been applied to this component, contradicting the fact that $G$ is a reduced instance.
    
    Now consider the set of type 1 components $D_1, D_2,\dots , D_r$ of $T \setminus A'$ attached in $T$ to the same node $a \in A'$. Dumas \etal~\cite{DPT23} showed that $W_a = W(D_1) \cup W(D_2) \cup \dots \cup W(D_r)$ is a trivially perfect module of $G$. 
    By \cref{lem:reducedmodule}, we have $ \vert W_a \vert  = c(k)$. There are at most $ \vert A' \vert  \leq 6k$ such sets $W_a$, thus the set of vertices of $G$ in bags of type 1 components is of size $O(k \cdot c(k))$. 
    
    Now consider the type 2 connected components $D$ of $T \setminus A'$ which have two neighbors in $T$. 
    Let $a_1$ and $a_2$ be these neighbors, one being the ancestor of the other, say $a_1$ is the ancestor of $a_2$. 
    Let $\{a_1, t_1,\ldots, t_l, a_2\}$ be the path from $a_1$ to $a_2$ in the tree.
    The component $D$ can be seen as a comb of shaft $(B_{t_1}, \dots, B_{t_l})$. 
    More precisely, by construction of the universal clique decomposition, $W(D)$ can be 
    partitioned into a comb $(C,R)$ of $H$: the critical clique decomposition of $C$ is $(C_1=B_{t_1},\dots, C_l=B_{t_l})$, and each $R_i$ corresponds to the union of bags of the subtrees rooted at $t_i$ which do not contain $B_{t_{i+1}}$, for $1 \leq i<l$, and to the union of bags of the subtrees rooted at $t_l$ which do not contain $a_2$, for $i=l$. Since $(C,R)$ was not affected by $F$, it is also a comb of $G$. Thus for each type 2 component $D$, $W(D)$ contains $p(k)$ vertices. 
    Since $T$ is a forest, it can contain at most $ \vert A' \vert  - 1 \leq 6k - 1$ such components in $T \setminus A'$. Therefore the set of bags containing type 2 connected components of $T \setminus A'$ contains $O(k \cdot p(k))$ vertices. \\ 
    It remains to bound the set of vertices of $G$ which are in bags of $A'$. The vertices corresponding to nodes of $A' \backslash A$ are critical cliques of $G$, and are hence of size at most $k+1$ by Rule~\ref{rule:borneCC-TP}. Thus the set of vertices in bags of $A' \backslash A$ is of size $O(k^2)$. 
    We conclude by showing a similar bound for vertices in bags of $A$.  
    Such vertices induce critical cliques in $H$ but not necessarily in $G$. However, note that in each such 
    critical clique the set of 
    vertices not affected by $F$ correspond to clique modules in $G$ (not necessarily maximal). Such vertices are contained in 
    exactly one critical clique of $G$ and have thus been 
    reduced by \cref{rule:borneCC-TP}. It follows 
    that the set of affected critical cliques of $H$ contains at most $2k + 2k \cdot (k+1)$ vertices.     
    Altogether we 
    obtain that $\size{V(G)} = O(k^2 + k \cdot (p(k)+c(k)))$ which concludes the proof of \cref{thm:strategy}. To obtain \cref{thm:result} we simply recall that \cref{lem:reducedmodule,lem:combs} imply that $c(k) = O(k)$ and $p(k) = O(k)$, respectively.
\end{proof}

\section{The deletion variant}
\label{sec:deletion}

As mentioned in the introduction, a quadratic vertex-kernel is known to exist for \textsc{Trivially Perfect Completion}~\cite{BBC+22}. 
The results presented \cref{sec:rules} can be adapted to prove that \textsc{Trivially Perfect Deletion} also admits a quadratic vertex-kernel by simply replacing any mention of ``editing'' by ``deletion''. 

More precisely, one can see that in order to prove the safeness of \cref{rule:dam,rule:shaft,rule:teeth}, 
the $k$-editing $\OPT{F}$ for the original graph that is derived from a $k$-editing $F'$ for the reduced instance only uses operations that were done by $F'$. In particular, if $F'$ only contains non-edges then so does $\OPT{F}$, meaning that it is a valid solution. Together with the fact that \cref{rule:compTP,rule:borneCC-TP} are safe for the deletion variant, we obtain the following. 

\begin{theorem}
    \label{thm:resultTPD} 
    \textsc{Trivially Perfect Deletion} admits a kernel with $O(k^2)$ vertices. 
\end{theorem}

We conclude by mentioning that \cref{thm:result} also holds for \TPC{} for the same reasons as the deletion variant. However, a vertex-kernel with $O(k^2)$ is already known for this problem~\cite{BBC+22}. Moreover, the constant factor on the number of vertices is smaller than the one of our kernel.

\section{Conclusion}
\label{sec:conclusion}

In this work we improved known kernelization algorithms for the 
\TPE{} and \textsc{Trivially Perfect Deletion} problems, providing a quadratic vertex-kernel for both of them. This
matches the best known bound for the completion variant~\cite{BBC+22}. 
Improving upon these bounds is an appealing challenge that may require 
a novel approach. 
On the other hand, it would be interesting to develop lower bounds for kernelization on such problems. Finally, even if the use of unaffected vertices in the design of reduction rules is common, its combination with the structural properties of trivially perfect graphs in terms of their maximal cliques allowed us to design stronger reduction rules.  
We hope that 
the approach presented in this work may lead to finding or improving kernelization 
algorithms for some related problems. 
Let us for instance mention the 
cubic vertex-kernel for \textsc{Proper Interval Completion}~\cite{BP13} and the quartic one for \textsc{Ptolemaic Completion}~\cite{CGP21} that might be appropriate candidates.  

\bibliographystyle{plain}
\bibliography{bibliography}

\end{document}